\newif\ifNOTES \NOTESfalse   
\newif\ifPOPL \POPLtrue   
\newif\ifweak\weakfalse   
\newif\iffull\fullfalse

\newif\iflong\longfalse 

\documentclass[fleqn]{llncs}

\usepackage{mathrsfs}
\usepackage{xspace}
\usepackage{mathpartir}
\usepackage{amsmath}
\usepackage{amssymb}

\usepackage{DSarrow} 
\input{commonmacros.sty}

\input{davidemacro.sty}

\begin{document}

\title{An abstract account of  up-to techniques for inductive behavioural relations}
\author{Davide Sangiorgi}

\institute{University of Bologna and INRIA}

\maketitle

\newcommand{\newpageDS}{}

 \begin{abstract} 
`Up-to techniques' represent enhancements of the coinduction proof method and are widely
used on coinductive behavioural relations such as bisimilarity.  Abstract formulations
of these coinductive techniques exist, using fixed-points or category theory.

A proposal has been recently put forward \cite{Sangiorgi22} for transporting the
enhancements onto the concrete realms of \emph{inductive} behavioural relations, i.e.,
relations defined from inductive observables, such as traces or enriched forms of traces.
The abstract meaning of such `inductive enhancements', however,  has not been explored.  In
this paper, we review the theory, and then propose an abstract account of it, using
fixed-point theory in complete lattices.
\end{abstract}

\section{Introduction}
 \label{s:intro}

The recent  paper  \cite{Sangiorgi22} presents 
 an attempt  at transporting 
the well-known theory of 
 enhancements for  \emph{coinductive}  behavioural relations
such as bisimilarity  (the so-called \emph{`up-to techniques'})
onto 
the realms of \emph{inductive} behavioural relations, i.e.,  relations defined from 
inductive
 observables, such as traces or enriched forms of traces
 \cite{Gla01}.
Abstract formulations
of the theories of  coinductive enhancements
exist, using 
fixed-points or category theory~\cite{PousS19}. 
However, 
the abstract
meaning of the `inductive enhancements'~\cite{Sangiorgi22}
 is  unclear, both because the observables are inductive and because
of the coinductive flavour of some of the constructions at the heart of the theory. 
In this paper we  give an overview of the concrete theory in 
 \cite{Sangiorgi22},  and then propose an abstract account 
of it, using  fixed-point theory in complete lattices.

The motivation  for transporting the bisimulation  enhancements 
 onto
 inductive  
behavioural  relations  is twofold. 
(We focus on 
the  bisimulation enhancements, 
as bisimilarity is  the most important 
 coinductive behavioural relation and its enhancements are, by far, the most widely used.)
First, the bisimulation enhancements  can be
very effective in proofs. The enhancements build on the bisimulation proof method,
whereby, 
to prove two terms bisimilar, one finds a bisimulation relation, 
 that is, a relation
invariant under the observables of the language (i.e., what can be
observed of the terms). 
In the enhancements, the invariance is allowed to hold `up-to' certain functions on the
relation.
Examples of such `up-to functions'  include  `up-to context', 
`up-to transitive closure',
`up-to bisimilarity', `up-to environment' \cite{PousS19}.
Most important,  the theories of enhancements
allow one to \emph{combine} enhancements 
so to obtain, for free, the soundness of more complex   enhancements
\cite{San98MFCS,pous07:clut,pous:lics16:cawu}. 
Sometimes  the
 enhancements seem essential to be able to carry out a proof:
  defining a full  bisimulation, with all needed pairs,  can be considerably
hard, 
let alone carrying out the necessary checks (see e.g., \cite{San93d}).  
This is frequent 
 in   languages for name
mobility stemming from the $\pi$-calculus 
 and in languages including higher-order features 
 such as $\lambda$-calculi.
Indeed, in these languages bisimilarity is hardly ever applied without
enhancements. 

The other major motivation for transporting the bisimulation enhancements onto inductive
behavioural relations has to do with
the criticism that has been raised against bisimilarity.
To begin with, 
 bisimilarity  is sometimes regarded as {too fine},  discriminating
processes that an external observer could not  tell apart. 
Thus, for instance, in the CSP
community failure equivalence  \cite{BrookesHR84} is used in place of bisimilarity. 
Another argument against bisimilarity is that it does not have a natural associated
{preorder}.  For instance, similarity~--- the `one-way' bisimilarity~---
or variants of it, do not
yield bisimilarity as their induced equivalence \cite{SanbookINTRO}. 
Further, 
 similarity as a behavioural preorder is often inadequate
because it does not respect deadlocks. 
Various  inductive  behavioural relations, both preorders and equivalences, have been put
forward and studied  
that improve on
 such  limitations: examples are preorders based on traces, failures,    ready sets, refusals,    may
 and must testing,  ready and failure traces,  see, e.g.,
 \cite{BrookesHR84,BrookesR84,DeHe84,Pnu85,Pomello85,BaetenBK87,OlderogH86,Phillips87,Hen88,Gla01}.  
We call these relations and their enhancements `inductive' because
their  observables are inductively defined, though this represents
an  abuse of terminology
as only the observables are inductive objects.

In abstract formulations of coinduction 
using
fixed-point theory in complete lattices, 
the coinductive objects are greatest fixed-points of monotone
functions.  
The coinductive enhancements can then be explained by means of 
composition of 
these functions with auxiliary functions 
(the 'up-to' functions). Appropriate conditions on the composition of the two functions
guarantee that the post-fixed points of the  composed function remain below the
initial coinductive object \cite{SanPous}.



In the abstract presentation of inductive enhancements  in this paper, 
a  difference 
is that  one reasons about 
  meets (i.e., greatest lower bounds) of  chains of  decreasing  points in
the lattice,
starting from the top element. Then, the functions of interest, the  
\emph{valid} functions for a given 
chain,  are those  whose post-fixed
points are below the points of the chain~-- they need not be functions that produce the
points of the chain 
by iteration of the  
application of the function to the top element. 
Valid functions  are then  enhanced by means of 
\emph{\ap
functions}. These are, intuitively functions that 
`respect the approximants' (i.e., the points in the chain). They
 are the analogous of the `up-to functions' of the coinductive setting.

In the concrete setting of 
inductive behavioural relations
 \cite{Sangiorgi22},
the inductive observables
are  described by means of 
{modal formulas}, 
precisely subsets of   the positive fragment of Hennessy-Milner logic
with the possible addition of  
 atomic  observables. 
A \emph{weight} 
can thus  be    associated to
 each
observable, intuitively  expressing  the depth of the nesting of actions 
in the behaviour
of a process  that may have be looked at when checking that observable. 
Approximants are then obtained by placing  bounds on the weights of the  observables employed.
The enhancements proposed in \cite{Sangiorgi22} are functions $\FF$ 
on relations that 
are  \emph{weight-preserving} for a given preorder $\preS$, i.e., whenever a relation
$\R$ complies with the $n$-th approximant  $\preSN n$ of the preorder,
 then also 
$\ff\R$  complies (i.e., 
${\R } \subseteq {\preSN n}$
implies  
${\ff \R} \subseteq {\preSN n}$). 
Such enhancements
are used within 
\emph{\quasiprogrS} between process relations.
A  \quasiprogr from  $\R$ to   $\S$ 
indicates a match in the immediate actions from pairs of processes in
$\R$  (akin to the match required in a simulation relation)
in which the derivatives (i.e., the processes resulting
from performing such actions)  are
in  $\S$. 
The soundness theorem states that, for any  function $\FF$ that is \WP
for a preorder $\preS$, 
a 
  \quasiprogr from $\R$ to $\FF(\R)$ 
allows us to conclude 
that the relation $\R$ is  contained in 
$\preS$.

When relating this concrete setting to the abstract one 
 using fixed-points in complete lattices  briefly
described above:
   the set of all process relations is the complete lattice;
 the sequence of weight-based approximants of a
behavioural preorder is a 
 chain of decreasing points in the lattice; the  preorder itself is the meet of the chain; 
the \quasiprogrS yield  valid functions; and the \WP functions are the \ap functions.

The theory in this paper, as well as that in  \cite{Sangiorgi22}, 
is  parametrised on a set of inductive observables, that is, a preorder. 
See  \cite{Sangiorgi22} for examples of  proofs of results about behavioural preorders   that 
are \emph{parametric} on  a  preorder, hence applicable to  a number of preorders.  
In this paper we only consider `strong' semantics, i.e., Labeled Transition Systems in
which  all actions are equally visible.   

\paragraph{Structure of the paper} Sections~\ref{s:strong} and~\ref{s:frs}
collect the main background.  
Section~\ref{s:cl} presents the abstract setting based on complete
lattices, post-fixed points, and chains of approximants.
Finally,  Section~\ref{s:concl}  reports the main related work and
some possible directions for future work.

     \newpageDS  

\section{Inductive observables and weights}
\label{s:strong}

In this section we recall the definitions of 
inductive observables and preorders from \cite{Sangiorgi22}, and their approximants using weights. 
We  work on ordinary  \emph{Labeled Transition Systems} (LTSs), i.e., 
triples  $(\pr, \Act, \longrightarrow )$ with domain $\pr$, set of {\em actions\/}
(or {\em labels}) $\Act$, and transition relation ${\longrightarrow } \subseteq  
{\pr \times \Act \times \pr}$. We use $P,Q$ and $R$
to range over $\pr$ and call them {\em processes};  $\mymu$ 
ranges  over
the  actions.
 As usual  $P \arr\mu Q$ stands for 
${(P, \mu , Q)} \in {\longrightarrow }$,
 to be interpreted as ``$P$ may   become
$Q $  by performing an action $\mu$''.  
We write 
$P \arr\mu $ if  there is $P'$  with $P \arr\mu P'$, and 
$P \notarr \mu$ if there is no $P'$ such that $P \arr\mu P'$.  
We use:  $n$ to range over the natural numbers; 
  letters $I$ and $J$  for  countable indexing sets
 in unions,  conjunctions, and similar; 
$i, j$ to range over
countable sets.

\subsection{Inductive  observables}
\label{ss:io}

The observables  are described by means of 
modal formulas. The operators  include the `diamond' $\act \mu. \obs$,
to detect the possibility of performing the action $\mu$, 
 a (possibly
infinitary) 'and', to 
permit
multiple observations, 
and a set of atomic  observables. 
Without the atomic observables, this is the positive fragment of Hennessy-Milner logic 
\cite{HeMi85}. 

\begin{definition}
\label{d:obs}
Let $\Act$ and $\AT$  be   (countable) sets of \emph{actions} 
and of  \emph{atomic observables}. 
A 
 \emph{family of observables  (over $\Act$ and $\AT$)}  is  a subset 
$\Obs$   
of  the formulas  inductively generated by the following grammar:
\iffull
\[ 
\begin{array}{rcll}
\obs & := &    \act \mu . \obs  
\hskip 2cm 
 &  \mbox{(\rn{act})} \\
& 
\midd 
& 
\bigwedge_{j\in J} \obsI j
 &  \mbox{(\rn{conj})} \\
&   
\midd 
& 
   \aat
  &  \mbox{(\rn{at})} 
s\end{array} 
 \] 
\else
\[ 
\begin{array}{rcll}
\obs & := &    \act \mu . \obs  
\midd 
\bigwedge_{j\in J} \obsI j
\midd 
   \aat
\end{array} 
 \] 
\fi
where
    $\mymu  \in \Act$,  $\aat  \in \AT$,  
$J$ is a countable set, and:
\begin{itemize}
\item (downward-closure)  if $\act\mu. \obs \in \Obs $ then also  
$ \obs \in \Obs $, and similarly if 
$\bigwedge_{j\in J} \obsI j \in \Obs$ then also $\obsI j \in \Obs$ for
each $j$.
\end{itemize} 

We write $\true$  
for the formula corresponding to 
the empty conjunction, and $\obs_1 \wedge \obs_2$ for a binary conjunction.    
We let $\Obs$ range over  families of observables,
and 
 $\obs,\obsZ$  over observables. 
\qedD\end{definition}

The observables of a family $\Obs$ 
are meant to be interpreted as predicates
over the processes  of an LTS over the same set of actions, 
 and accordingly we write $\sat P \obs$ when the observable $\obs$ holds  on the process $P$.
The interpretation is defined
structurally over the observables, writing $\set\aat$ for the set
of processes on which the atomic observable $\aat$ holds:  
\[ 
\shortinfrule{at}{ P \in  \set{\aat}} 
{\sat P  \aat }
\hskip 1cm 
\shortinfrule{act}{ 
 P \arr\mymu  P'  \andalso   \sat {P'}{ \obs} 
}
{\sat P {   \act \mu . \obs} }
\hskip 1cm 
\shortinfrule{conj}{
\mbox{ 
$\sat P {\obsI j} $  \hskip .5cm for all $j \in J$
}}
{\sat P    {\bigwedgeD_{j\in J} \obsI j}
}
\]
The observable 
\true\  is   satisfied by all processes.
An example of  atomic observable
is 
   the refusal  $\refu \mu$,
which  holds of a process $P$ and an action $\mu$  if $P \notarr \mu$.
All atomic  observables 
should be   
\emph{local}:  
they can be checked  by looking only
at the immediate outgoing transitions emanating from 
a process. That is, for each such observable, say $\aat$, there are sets
of actions $A$ and $B$ such that for all process $P$, we have 
$ P \in  \set{\aat}$ 
iff  ( $P \arr \mu$ for all $\mu \in A$, and $P \notarr
\mu$ 
for all $\mu \in B$). 
 For instance, for the refusal $\refu \mu$, we have $A = \emptyset $ and $B = \{ \mu\}$.

\begin{definition}[preorder induced by $\Obs$]
\label{d:ie}
A  family $\Obs$  of observables   induces  a preorder $\preSM\Obs$ on
the processes, 
 whereby $ P \preSM\Obs Q $ if
$\sat P \theta $ implies
$\sat Q \theta $  for all $\theta \in \Obs$.
\iffull 
In this case we refer to  $\Obs$ as \emph{the observables for
$ \preSM\Obs$}.  
\fi
\qedD\end{definition}

The   equivalence induced by a preorder is obtained by 
 adding the  symmetric 
requirement on pairs of related  processes. 
All results below for preorders can be extended to equivalences in the expected manner. 
 We will  ignore equivalences in the remainder of the paper.


By varying 
the sets of observables, one 
can  capture  all inductive preorders and equivalences in van Glabbeek's
spectrum \cite{Gla01}
(collecting the main such relations in the literature), 
namely: 
\emph{trace}, \emph{failure}, \emph{failure trace}, \emph{ready}, \emph{ready trace}
preorders and equivalences \cite{BrookesHR84,BrookesR84,OlderogH86,Pnu85,Pomello85,BaetenBK87,Den87}. 
Assuming 
{image-finiteness} of the LTSs, meaning  that 
 for all $P$ and $\mu$,
 the set $\{ {P'}  \st {P}  \arr\mu {P'} 
\}$ 
is finite, the set of such preorders also include the  
\emph{may}, \emph{must}, and \emph{testing} preorders  \cite{DeHe84,Hen88}, 
and the  \emph{refusal preorder}  \cite{Phillips87}.
We refer to \cite{Sangiorgi22} for the details of the characterisations. 
The operators for  observables in 
Definition~\ref{d:obs}   also capture
  some  coinductive relations, such
as  simulation, and refinements of it
such as ready simulation \cite{Sangiorgi22}.

  \newpageDS \subsection{Weights}
\label{ss:we}

We are interested in observables that can be measured, that is,
 whenever $\sat P \obs $ holds then  a weight can be assigned to a proof 
of $\sat P \obs $. 
Given  
an LTS
 $\LL$   and a family $\Obs$ of observables over the same set of actions, 
a process  $P$  of  $\LL$,   and an observable $\obs \in \Obs$, 
 an assertion  
$\satN n P {\obs}$ 
 holds  if 
$\sat  P {\obs}$ 
and 
  the depth of nesting of diamond
operators in $\obs$ is not greater than $n$ (See \cite{Sangiorgi22} for a
definition in terms of rules).
For instance, 
we have $\satN n  P{\act{a}. \act b.\true} $, for 
any $n\geq 2$, whereas 
 $\satN 1  P{\act{a}. \act b.\true}$ does not hold.

\begin{definition}
\label{d:obsWobs}
A \obssys $\Obs$ and the corresponding preorder 
$\preSM\Obs$
are 
  \emph{measurable} 
if, for any LTS (over the same set of actions) 
and process $P$ of the LTS, and for any observable 
$\obs \in \Obs$, 
whenever $\sat P \obs$ then  
 also $\satN n P \obs$ holds, for some 
$n \geq 0$.
\qedD\end{definition}

The introduction of weights allows us to define approximations of the behavioural preorders.

\begin{definition}[stratification]
\label{d:strati}
Two processes $P,Q$ are in the \emph{$n$-th approximant of $\preSM \Obs$}, written
 $P \preSMN \Obs n Q$, if 
 $\satN n  P {\obs}$  implies  
$\sat Q \obs $, for all $\obs$.
 \qedD\end{definition}

In
the above definition,
$\sat Q \obs $  could be replaced by $\satN n  Q { \obs} $. 
Relations $ \preSMN \Obs n $ are preorders.

\begin{lemma}
\label{l:stra}
If
  $
\Obs$ is measurable,
then 
${\preSM \Obs} = {\cap_n \preSMN \Obs n}$.
\end{lemma}

\iffull
\begin{proof}
Immediate, from the definition of measurable,  Lemma~\ref{l:PQn}, and Remark~\ref{r:PQn}.   
\end{proof}   
\fi

All preorders discussed  at the end of  Section~\ref{ss:io}   are measurable.
In the remainder, it is intended that 
 all \obssys and all preorders
 are measurable.

     \newpageDS \section{Functions for inductive enhancements}
\label{s:frs}

We recall here the concrete definitions of the 
functions for inductive enhancements, namely the \WP functions \cite{Sangiorgi22}.
Unless otherwise stated, a \emph{relation} is meant to be a   (binary) relation on
processes.
We let $\R$ and $\S$ range over  relations.
The
union of  
relations $\R$ and $\S$ is  
$\RR \cup \SS$,  
 and their composition is 
$\R\S$   (i.e.,   
 $ \rmemb P  {\, \R  \S} { P'}$ holds if 
 for some $P''$,
 both 
$\rmemb P \RR {P''}$ and $\rmemb{P''}\SS{ P'}$ hold).
  We often use the infix notation
for relations; hence   ${P} \RR {Q}$  means  $(P,Q) \in \R$.  

The definition below of \quasiprogr is
essentially the ``one-way''
version of the  progressions in the theory of
enhanced coinduction \cite{San98MFCS}.  
The converse clause is missing, 
as
we are  dealing with preorders,
rather than equivalences as in enhanced coinduction. 

\begin{definition}[\quasiprogr]
  \label{d:diag}
  Given two relations $\R$ and $\S$, we say that $\R$ \emph{\quasiprogrs
    to} $\S$, written $\Lprog\R\S$, if $ P\RR Q$ implies:
  \begin{itemize}
  \item whenever $P \arr\mymu P'$, there exists $Q'$ such that $Q\arr\mymu Q'$
    and $P'\SS Q'$.
\qed
\end{itemize} 
\end{definition}

We need \quasiprogrS
of the form  $\Lprog\R{\app\qff\R}$
where $\qff$ 
 is  a  function 
on process relations.
Below, 
$\qff$ and $\qg$ range over such functions.
To handle
  atomic observables such as 
refusals,  the following
  \compatibility requirement is needed.

\begin{definition}[\compatibility]
\label{d:compa}
A relation $\R$ is \emph{\compatible \cfor $\Obs$} (or simply
\emph{\compatible}, if $\Obs$ is clear) if, 
for all atomic observables  $\aat$ in $\Obs$ and 
 for all    $P \RR Q$,
\iffull
\begin{itemize}
\item
\fi
$\sat P \aat$  implies
$\sat Q \aat $.
\iffull
\end{itemize} 
\fi
\qedD\end{definition} 


Separating \compatibility  and (semi-)progression
allows us
 to maintain the standard notion of progression in
the literature, which only refers to action
transitions. 
Moreover, in this way
 progressions and \quasiprogrS  
are defined   without reference to any family of  
 observables.

\begin{definition}[soundness]
A function  $\qff$  is  {\em  sound for $\preSM \Obs$}  if
$\Lprog\R{\ff\R}$  implies  ${\R }\subseteq  {\preSM \Obs}$, 
for any \compatible $\R$.
\qedD\end{definition}
 
Not all functions  are sound. An example is  the function that maps
every relation onto the universal relation (the set of all pairs of processes).
One  then looks for a class of sound functions
 for which membership
is easy to check,  
which includes interesting functions,  and which  satisfies
interesting properties.
For this  the notion of \WP function is proposed \cite{Sangiorgi22}. 

\begin{definition}[\WP function]
\label{d:observable_pres}
A   function $\qff$  \emph{\OP on $\preSM \Obs$}
  when, 
for all $n$ and for all  ${\R}$, if 
$\R \subseteq {\preSMN \Obs n}$, then also ${ \ff \R}
\subseteq {\preSMN \Obs n}$.  
\qedD\end{definition}

Any \WP function is indeed sound, as stated in the following theorem.  A direct proof of the theorem is
presented   in   \cite{Sangiorgi22}; in this paper we will derive it from the  fixed-point
theory of Section~\ref{s:cl}.

\begin{theorem}[soundness of \WP functions for  preorders]
\label{t:op}
If 
$\qff$ \OP\ on $\preSM \Obs$, then $\qff$ is sound for  $\preSM \Obs$. 
\end{theorem}

Thus the proof technique stemming from the theory above goes as
follows. 
Let
$\preSM \Obs$ be the preorder resulting from the   
 set of observables  $ \Obs$. 
Take a function
$\qff$ that \OP.
Then for any 
 relation  $\R$ that  is \compatible \cfor $\Obs$, whenever 
$\Lprog\R{\ff\R}$,  we also have   ${\R }\subseteq  {\preSM \Obs}$.

A number of examples of \WP functions are  reported in   \cite{Sangiorgi22}:
 the identity function, various useful constant functions, 
 the transitive-closure function, the closure under contexts (for specific languages). 
Moreover, the set of \WP functions is closed under useful constructors such
as  
function composition,  union, chaining (that yields  relational composition); these allow
one to derive a number of more sophisticated \WP functions, e.g.,  
 the transitive-closure function, 
the closure under $\preSM \Obs$ 
 (analogous of the `up-to bisimilarity' enhancement for
bisimilarity),  
the closure under contexts
and $\preSM \Obs$ 
 (analogous of the `up-to context and bisimilarity' enhancement for
bisimilarity). 

     \newpageDS

\section{Complete lattices and chains of approximants}
\label{s:cl}

In this section we present the abstract account for the enhancements of inductive
behavioural relations, using complete lattices, chains of approximants, and post-fixed
points. 

\subsection{Post-fixed points and chains}
\label{ss:pc}

We recall that  a \emph{poset} $L$
is a non-empty set
equipped with  a relation $\leq$ on its elements that is
a partial order (that is, reflexive,  antisymmetric,  and
transitive). 


\begin{definition}
\label{d:poset_FUN}
 Let  $\FF$ be an endofunction
on  a poset $L$.
 \begin{itemize}
 \item
 $\FF$ 
 is \emph{monotone}
 if $x \leq y$ 
 implies $\FF(x) \leq \FF(y)$, for all $x,y$. 

 \item 
 An element  
 $x$ of the  poset 
 is a 
 \emph{post-fixed point} of $\FF$
  if
 $x \leq \FF(x)$.

\item 
 A
 \emph{fixed point} of $\FF$  is an element 
$x$ with    $\FF(x) = x$.
  \qed
 \end{itemize}
\end{definition} 

The \emph{greatest lower bound} of a subset $S \subseteq L$ 
is also called the 
\emph{meet} of $S$; dually,  
the
\emph{least upper bound} of $S$
is also called the 
\emph{join} of $S$.
We write  $\cup_{x_n} x_n$ for the join of  a subset $\{x_n\}_n$ of elements of  the poset, and 
similarly $\cap_{x_n} x_n$  for its meet.

\begin{definition}[Complete lattice]
A \emph{complete lattice} is  
a poset   with all joins (i.e., all the 
 subsets of the  poset  have a  join). \qed
\end{definition}
The above  implies that a {complete lattice} has also all meets. 
Further,  it implies
that there are  bottom and top elements,
which will be   indicated
 by
$\bot$ and $\top$, respectively.
The key ingredient for our abstract theory of up-to techniques are 
 {chains of approximants} in  complete lattices.

\begin{definition}
\label{d:ca}
In a complete lattice, a \emph{chain of approximants}  is a  sequence of points
$\{x_n\}_n$, $n\geq 0$ in the lattice with  $x_0 = \top$ and $x_{n+1}\leq x_n$, for all $n$. 
\end{definition}  
We wish to reason about meets of  chains of approximants.  As we shall see,  these meets represent inductive
behavioural relations, obtained as the intersection of their approximants. 
We sometimes simply call \emph{chain} a  chain of approximants.

\begin{definition}
\label{d:valid}
An endofunction $\B$ on a complete lattice $L$
 is \emph{valid for a chain of approximants} $\{x_n\}_n$ in $L$ if, for all
$y \in L$ and $ n \geq 0$, 
if  \mbox{ $y \leq x_n$  then $\B(y) \leq x_{n+1}$. } 
\end{definition}  

Intuitively, $\B$ is valid for a chain
  if an  application of $\B$ allows us to move (at least) one step
down  in the  chain. 
Below, for readability  we often omit reference to the complete lattice. 
It is intended that all points and functions are points and endofunctions on a  lattice.

Valid functions for a chain allow us to reason about the meet of  
the chain.

\begin{theorem}
\label{t:valid}
If $\B$ is valid for the chain  $\{x_n\}_n$, then for any post-fixed point $z$ of $\B$ we  have 
$z \leq  \cap_n x_n$.
\end{theorem} 

\begin{proof}
We show that $z \leq   x_n$, for each $n$, using induction on $n$. 
For $n=0$, this follows from  $x_0=\top$. 
Otherwise, suppose $z\leq x_n$; by validity of $\B$ we derive $\B(z) \leq x_{n+1}$. As $z$ is
a post-fixed point of $\B$, we have $z\leq \B(z)$. We can thus conclude that also $z\leq 
x_{n+1}$.
\qed \end{proof}

We seek ways of enhancing valid functions, i.e., obtain other valid functions (with the
intention of obtaining, concretely,  
'larger' functions). For this we use the \ap functions.

\begin{definition}
\label{d:ap}
An endofunction $\FF$ 
 is \emph{\ap for a chain of approximants} $\{x_n\}_n$ 
if, for all
$y 
$ and $ n \geq 0$, 
if  \mbox{ $y \leq x_n$  
then   also $\FF(y) \leq x_{n}$. } 
\end{definition}  

Intuitively, application of \ap functions respects the level of approximation.
We write $ \comp {\FF_1}{\!\FF_2}$ 
for 
the composition of the two functions, where $(\comp{ \FF_1 }{\!\FF_2}) (z) =
\FF_1(\FF_2(z))$). 

\begin{theorem}
\label{t:valid_comp}
If $\B$ is valid,  and $\FF$ is \ap,  for the same chain $\{x_n\}_n$, 
then also $ \comp {\B}{\!\FF}$ 
is  valid for  $\{x_n\}_n$, which means that  $\B$ is indeed valid.
\end{theorem} 

\begin{proof}
Suppose $y \leq x_n$, for some $n$ and $y$. As $\FF$ is \ap, we have $ \FF(y) \leq x_n$; 
and since $\B$ is valid, $\B(\FF(y)) \leq x_{n+1}$.
\qed
\end{proof}

The set of \ap functions is closed under composition, join, and contains a few useful
constant functions.

\begin{proposition}
\label{p:ap}
For any chain of approximants $\{x_n\}_n$,

\begin{itemize}
\item the identity function is \ap;
\item for any $y$ with $y \leq \cap_n x_n$ the constant function mapping each  point onto
  $y$ is \ap;
\item if $\FF_1$ and $\FF_2 $ are \ap then so is $\comp{\FF_1}{\FF_2}$; 

\item if $\{ \FF_j\}_j$ is a set of \ap functions, then $\cup_j \FF_j$, i.e., the function
  that maps a point $z$ onto $\cup_j \FF_j(z)$ is \ap too.
\end{itemize} 
\end{proposition} 

\subsection{Monotone functions}
\label{ss:mono}

The conditions of validity and \ap can be simplified if we require the functions to be
monotone.

\begin{proposition}
\label{p:mon-val}
If $\B$ is monotone, then, for any 
 chain of approximants $\{x_n\}_n$,
$\B$ is valid for $\{x_n\}_n$ iff, for all $n$, we have $ \B(x_n) \leq  x_{n+1}$.
\end{proposition} 
\begin{proof}
If $\B$ is valid, then from $x_n \leq x_n$ we infer 
$ \B(x_n) \leq  x_{n+1}$.
Conversely, if $y \leq x_n$, then by monotonicity we derive 
$\B(y) \leq \B(x_n)$; then, from  $ \B(x_n) \leq  x_{n+1}$, also 
$\B(y) \leq x_{n+1}$.
\qed \end{proof}                

A similar result holds for \ap functions:
\begin{proposition}
\label{p:mon-ap}
If $\FF$ is monotone, then, for any 
 chain of approximants $\{x_n\}_i$,
$\FF$ is \ap for $\{x_n\}_n$ iff, for all $n$, we have $ \FF(x_n) \leq  x_{n}$.
\end{proposition}

\subsection{From complete lattices to behavioural relations}
\label{ss:clbr}
We now show that the results recalled in Section~\ref{s:frs} are instances of the
abstract setting based on fixed-points in Sections~\ref{ss:pc} and \ref{ss:mono}.

The complete lattice to be used 
in the case of behavioural relations on a set $\pr$ of processes
 is obtained from the powerset construction on   $\pr \times \pr$. 
Thus
  $\emptyset $ (the empty set) and $\pr \times \pr$
are the bottom  and top elements;  join is given by set union, and
meet by set intersection.  
We write $\clpr$ for this complete lattice.

We recall that the  behavioural preorders of  Section~\ref{s:strong} are
obtained as the meet of their approximants, as for each set   $\Obs$ of observables, we
have ${\preSM \Obs} = {\cap_n \preSMN \Obs n}$ (Lemma~\ref{l:stra}). 
The approximants 
$ \{ \preSMN \Obs n\}_n$ form a chain of approximants (Definition~\ref{d:ca}), as
   $\preSMN \Obs {0} = \Pr \times
\Pr$ and,   for
each $n$, we have  ${\preSMN \Obs {n+1}} \subseteq {\preSMN \Obs {n}}$.


Next, we need to express the
semi-progressions  of Section~\ref{s:frs}
as  functions $\sobs \Obs$ on the lattice, for a given set $\Obs$ of observables.
The only difference with respect to the
ordinary functional for simulation is the compliance condition: 
\[ 
\sobs \Obs (\R) \defi
\{ 
\begin{array}[t]{l}
(P,Q) \st  \\
\mbox{1. whenever $P \arr\mymu P'$, there exists $Q'$ s.t.\
  $Q\arr\mymu Q'$ and $P'\RR Q'$;}\\
\mbox{2. for all atomic observables  $\aat$ in $\Obs$ 
if $\sat P \aat$  then also 
$\sat Q \aat $.}  \}
\end{array} 
\]

\begin{theorem}
\label{t:sobs_valid}
For all $\Obs$,  function 
$\sobs \Obs$ is valid  for $ \{ \preSMN \Obs n\}_n$.
\end{theorem}

\begin{proof}
We omit $\Obs$, abbreviating
$\preSMN \Obs n$ as 
$\preSN n$, and $\sobs \Obs$ as $\S$. 
The function $\S$ is monotone. Hence, by  Proposition~\ref{p:mon-val}, to show that 
$\S$ is valid for $ \{ \preSMN {} n\}_n$ it suffices to show that, 
for any $n$, we have 
${\S( \preSMN {} {n}) }\leq { \preSMN {} {n+1}}$.

Thus take any $n \geq 0$, and 
suppose $(P,Q) \in \S( \preSMN {} {n})$. We have to show that also 
$(P,Q) \in \ \preSMN {} {n+1}$.
For this, assume 
$\sat  P \obs$ and the depth of nesting of diamond operators in $\obs$ is not greater than $n+1$. We check that also 
$\sat  Q \obs$ holds. 
We proceed by rule induction on the proof of  $\sat  P{ \obs$}, using a case  analysis
on the last rule applied.

\begin{description}
\item[Rule \trans{at}]  For atomic observables, 
$\sat  Q \obs$ holds by the compliance condition on $\S$ (clause (2) in the definition of $\S$).

\item[Rule \trans{conj} ]
If the last rule applied is
\[
\shortinfrule{}{ 
 \sat  P{\obsI j}
 \mbox{ for     each $j$}
  } 
{\sat P      {\bigwedgeD_{j\in J} \obsI j}
}\]
then since, for each $j$, by the downward-closure property on $\Obs$ we have  
$\obsI j \in \Obs$. 
We can use  rule induction to infer
$\sat Q {\obsI j}  $, and then  conclude 
$\sat Q {\bigwedgeD_{j\in J} \obsI j}$.

\item[Rule \trans{act}]
Suppose that the last rule applied is
\[
\shortinfrule{}{ 
 P \arr\mymu  P'  \andalso   \sat {P'}{ \obs} 
}
{\sat P {   \act \mu . \obs}
}
\]
where,
by the downward-closure property on $\Obs$ we have  
$\obs \in \Obs$, 
and
 the depth of diamond operators in $\obs$ is not greater than $n$; 
that is, $ \satN {n} {P '}{\obs   }$ holds. 
From $ P \arr\mymu  P'$  
and  clause (1) of the definition  of $\S$, 
we infer that $Q \arr \mymu Q'$ with 
${P'} \:\preSN {n}\:  {Q'}$, for some $Q'$.  
From ${P'} \:\preSN {n}\:  {Q'}$ and  $\satN n  {P'}{ \obs}$, by
definition of 
$\preSN {n}$,
 we obtain $\sat{Q' }{ \obs }$. 
In summary, we can infer
\[
\shortinfrule{}{ 
 Q  \arr\mymu  Q'  \andalso   \sat {Q'}{  \obs}
}
{\sat Q   \act \mu . \obs
}\]
\qed
\end{description}
 \end{proof}

Finally, we notice that, on the lattice $\clpr$,  the definition of weight-preserving
function matches that of \ap functions. 
That is, a function is  
\OP on $\preSM \Obs$ precisely when it is 
\ap for  
$ \{ \preSMN {\Obs} n\}_n$.  

We thus derive the soundness of 
\WP
functions, i.e., Theorem~\ref{t:op}, stating that 
if 
$\qff$ \OP\ on $\preSM \Obs$, then $\qff$ is sound for  $\preSM \Obs$. 

 \begin{proof}[of Theorem~\ref{t:op}]
Soundness of  $\qff$   for  $\preSM \Obs$ means that any 
post-fixed point of 
 $\comp 
{\sobs \Obs} \FF$ is contained in  $\preSM \Obs$.

Since   \WP on $\preSM \Obs$ is the same as \ap for  
$ \{ \preSMN {\Obs} n\}_n$, by   Theorems~\ref{t:sobs_valid} and~\ref{t:valid_comp}, 
 function  $\comp 
{\sobs \Obs} \FF$ is 
 valid  for $ \{ \preSMN \Obs n\}_n$. Therefore for any post-fixed point $\R$ of 
 $\comp 
{\sobs \Obs} \FF$, we have ${\R} \subseteq \cap_n \preSMN \Obs n = \preSM \Obs$.
\qed \end{proof} 

Similarly, properties of 
\WP functions, such as closure under composition and union, can be derived from the
properties of \ap functions in the abstract setting (e.g., Proposition~\ref{p:ap}).


 \newpageDS\section{Concluding remarks}
\label{s:concl}

\subsection{Related work}

The idea of developing a  theory of enhancements, for the bisimulation proof method, is
put forward in  \cite{San98MFCS}, introducing the notion of progression 
(roughly the symmetric version  of the semi-progression of Section~\ref{s:frs})
and of  respectful
function (roughly, functions that preserve progressions).  
Later  \cite{pous07:clut,pous:lics16:cawu}
the theory has been both    generalised  to the coinduction proof method, using
complete lattices, and refined, first using a variant of  respectfulness called
compatibility, 
 then  focusing on the largest 
 compatible function, called the
\emph{companion} (that
also  coincides with the largest respectful function). 
The companion can also be obtained using Kleene's
construction of the greatest fixed-point~\cite{ParrowWeber16}.
The theory of enhancements for inductive behavioural relations \cite{Sangiorgi22}
follows the schema for 
the bisimulation  enhancements based on  {respectfulness} and  
{compatibility} \cite{SanPous}.
In contrast, the
 chains of approximants and the \ap functions, on which the abstract theory in Section~\ref{s:cl} is based,
closely resemble the approximations of the greatest fixed-point of monotone functions
from the top element of  the  lattice
using Kleene's
construction, and the corresponding account of `up-to' techniques, as studied in papers such
as~\cite{pous:lics16:cawu,ParrowWeber16,SchaferS17,pr:fossacs17:codensity}.

Abstract formulations of the meaning of coinductive  enhancements have also been given using category
theory. 
The main technical tools are 
final semantics,  coalgebras,  
spans of coalgebra homomorphisms, fibrations, and 
 corecursion schemes; see, e.g.,  
\cite{pr:fossacs17:codensity,bpr:calco17:monoidal:company,RotBBPRS17}.
For  more details and references on coinductive  enhancements, we refer to the 
 technical  survey \cite{SanPous} 
 and to the  historical  review~\cite{PousS19}.

Proof techniques for (weak) bisimilarity akin to the up-to techniques of coinductive
enhancements have also been studied using 
 \emph{unique solutions of equations} and 
special inequations called
\emph{contractions}  \cite{DurierHS19,Sangiorgi17}.
In particular, these techniques offer  one the power of some of the most  important 
 bisimulation enhancements, namely
 `up-to context'.
 Moreover, the authors  transport them  onto 
inductive equivalences and preorders such as 
trace equivalence and trace preorder.

\subsection{Further developments}

We have given an interpretation, using complete lattices and
fixed-points,  of the enhancements and the `up-to' techniques for inductive behavioural
relations \cite{Sangiorgi22}.

In the paper, we have confined ourselves to `strong' semantics,
i.e., to transition systems where all actions are treated in the same manner.
It would be useful to 
 extend the setting
to 
 `weak' forms of behavioural relations,
in which a special action denoting internal activity is  partly or completely ignored.  
Theories of weak coinductive enhancements are rather 
more involved than the strong  theories. 
For instance,
 auxiliary relations such as \emph{expansion}  \cite{SaMi92} and
 \emph{contraction} \cite{Sangiorgi17}  are usually needed.    
Similar issues show up in 
 enhancements     for inductive behavioural relations.
For instance, two forms of weights have been  proposed  
  \cite{Sangiorgi22}, 
 distinguishing the contribution of
internal and visible actions,  with their relative advantages and disadvantages. 
In addition, some weak
behavioural relations make use of state predicates such as \emph{stability}, which do not appear in
the strong case.

A challenging direction might be to study
lifting of the theory presented in this paper to a  probabilistic setting and 
 labelled Markov processes. 
We would like  also to  examine
abstract settings capable of accommodating behavioural relations
defined from both \emph{inductive} and \emph{coinductive} observables. 
Examples of coinductive observables that may appear together with the
inductive ones are
infinite traces and
divergence.





\subsection*{Acknowledgments}
Thanks  to 
 the anonymous referees for 
  useful observations and comments on the content of the paper.  
This work has been partly supported by the
 MIUR-PRIN project 
`Resource Awareness in Programming: Algebra, Rewriting, and
Analysis' (RAP, ID P2022HXNSC).

\subsection*{Personal Note}
I am very glad and honoured to have paper in a volume dedicated to
Rocco De Nicola.  This also gives me an opportunity for expressing my
gratitude to him. Not only for many fruitful scientific discussions,
but also (and possibly even more!)  for his constant and generous
friendship, including frequent pieces of personal advice, beginning
with the recommendation to go to Edinburgh for carrying out a PhD! 
Furthermore, meeting Rocco has always been a guarantee for joyful and
pleasant moments: conversation would quickly and invariably take a
broad span (from science, to life, and to politics...), with Rocco
brilliantly seizing all opportunities for jokes and laughter.
I look forward to having many more such moments in front of us!

\bibliographystyle{splncs03} 
\bibliography{DSbib,OTHERSbib} 

\end{document} 

\bibliographystyle{splncs03} 
\bibliography{DSbib,OTHERSbib}

\end{document}